\def\year{2019}\relax
\documentclass[letterpaper]{article} %DO NOT CHANGE THIS
\usepackage{aaai19}  %Required
\usepackage{times}  %Required
\usepackage{helvet}  %Required
\usepackage{courier}  %Required
\usepackage{url}  %Required
\usepackage{graphicx}  %Required

\usepackage{latexsym}
\usepackage{amsmath}
\usepackage{amsthm}
\usepackage{amssymb}
\newtheorem{theorem}{Theorem}

\newtheorem{lemma}{Lemma}

\begin{document}
\title{Solving Imperfect-Information Games \\via Discounted Regret Minimization}
\author{Noam Brown\\
	Computer Science Department\\
	Carnegie Mellon University\\
	noamb@cs.cmu.edu
	\And
	Tuomas Sandholm\\
	Computer Science Department\\
	Carnegie Mellon University\\
	sandholm@cs.cmu.edu
}
\maketitle
\vspace{-0.05in}
\begin{abstract}
\vspace{-0.05in}
\emph{Counterfactual regret minimization (CFR)} is a family of iterative algorithms that are the most popular and, in practice, fastest approach to approximately solving large imperfect-information games. In this paper we introduce novel CFR variants that 1) discount regrets from earlier iterations in various ways (in some cases differently for positive and negative regrets), 2) reweight iterations in various ways to obtain the output strategies, 3) use a non-standard regret minimizer and/or 4) leverage ``optimistic regret matching''. They lead to dramatically improved performance in many settings. For one, we introduce a variant that outperforms \emph{CFR+}, the prior state-of-the-art algorithm, in every game tested, including large-scale realistic settings. CFR+ is a formidable benchmark: no other algorithm has been able to outperform it.
Finally, we show that, unlike CFR+, many of the important new variants are compatible with modern imperfect-information-game pruning techniques and one is also compatible with sampling in the game tree.
\end{abstract}

\vspace{-0.02in}
\section{Introduction}
\vspace{-0.01in}
Imperfect-information games model strategic interactions between players that have hidden information, such as in negotiations, cybersecurity, and auctions. A common benchmark for progress in this class of games is poker. The typical goal is to find an (approximate) equilibrium in which no player can improve by deviating from the equilibrium.

For extremely large imperfect-information games that cannot fit in a linear program of manageable size, typically iterative algorithms are used to approximate an equilibrium. A number of such iterative algorithms exist~\cite{Nesterov05:Excessive,Hoda10:Smoothing,Pays14:Interior,Kroer15:Faster,Heinrich15:Fictitious}. The most popular ones are variants of \emph{counterfactual regret minimization (CFR)}~\cite{Zinkevich07:Regret,Lanctot09:Monte,Gibson12:Generalized}. In particular, the development of \emph{CFR+} was a key breakthrough that in many cases is at least an order of magnitude faster than vanilla CFR~\cite{Tammelin14:Solving,Tammelin15:Solving}. CFR+ was used to essentially solve heads-up limit Texas hold'em poker~\cite{Bowling15:Heads-up} and was used to approximately solve heads-up no-limit Texas hold'em (HUNL) endgames in \emph{Libratus}, which defeated HUNL top professionals~\cite{Brown17:Superhuman,Brown17:Safe}. A blend of CFR and CFR+ was used by \emph{DeepStack} to defeat poker professionals in HUNL~\cite{Moravcik17:DeepStack}.

The best known theoretical bound on the number of iterations needed for CFR and CFR+ to converge to an $\epsilon$-equilibrium (defined formally in the next section) is $O(\frac{1}{\epsilon^2})$~\cite{Zinkevich07:Regret,Tammelin15:Solving}. This is asymptotically slower than first-order methods that converge at rate $O(\frac{1}{\epsilon})$~\cite{Hoda10:Smoothing,Kroer15:Faster}. However, in practice CFR+ converges much faster than its theoretical bound, and even faster than $O(\frac{1}{\epsilon})$ in many games.

Nevertheless, we show in this paper that one can design new variants of CFR that significantly outperform CFR+. We show that CFR+ does relatively poorly in games where some actions are very costly mistakes (that is, they cause high regret in that iteration) and provide an intuitive example and explanation for this. To address this weakness, we introduce variants of CFR that do not assign uniform weight to each iteration. Instead, earlier iterations are discounted. As we show, this high-level idea can be instantiated in many different ways.
Furthermore, some combinations of our ideas perform significantly better than CFR+ while others perform worse than it. In particular, one variant outperforms CFR+ in every game tested.

\vspace{-0.02in}
\section{Notation and Background}
\vspace{-0.01in}
\label{sec:background}
We focus on sequential games as the most interesting and challenging application of this work, but our techniques also apply to non-sequential games.
In an imperfect-information extensive-form (that is, tree-form) game there is a finite set of players, $\mathcal{P}$. ``Nature'' is also considered a player (representing chance) and chooses actions with a fixed known probability distribution.
A state $h$ is defined by all information of the current situation, including private knowledge known to only a subset of players. $A(h)$ is the actions available in a node and $P(h)$ is the unique player who acts at that node. If action $a \in A(h)$ leads from $h$ to $h'$, then we write $h \cdot a = h'$. $H$ is the set of all states in the game tree.
$Z \subseteq H$ are terminal states for which no actions are available. For each player $i \in \mathcal{P}$, there is a payoff function $u_i: Z\rightarrow \mathbb{R}$. 
We denote the range of payoffs in the game by $\Delta$. Formally, $\Delta_i = \max_{z \in Z} u_i(z) - \min_{z \in Z} u_i(z)$ and $\Delta = \max_{i \in \mathcal{P}} \Delta_i$.

Imperfect information is represented by {\em information sets} (infosets) for each player $i \in \mathcal{P}$. For any infoset $I$ belonging to player~$i$, all states $h, h' \in I$ are indistinguishable to player~$i$. Every non-terminal state $h \in H$ belongs to exactly one infoset for each player~$i$. The set of actions that may be chosen in $I$ is represented as $A(I)$. We represent the set of all infosets belonging to player $i$ where $i$ acts by $\mathcal{I}_i$.

A strategy $\sigma_i(I)$ is a probability vector over actions for player $i$ in infoset $I$. The probability of a particular action $a$ is denoted by $\sigma_i(I,a)$. Since all states in an infoset belonging to player $i$ are indistinguishable, the strategies in each of them are identical. Therefore, for any $h \in I$ we define $\sigma_i(h,a) = \sigma_i(I,a)$ where $i = P(h)$.
We define $\sigma_i$ to be a strategy for player $i$ in every infoset in the game where player $i$ acts.
A strategy profile $\sigma$ is a tuple of strategies, one per player. The strategy of every player other than $i$ is represented as $\sigma_{-i}$. $u_i(\sigma_i, \sigma_{-i})$ is the expected payoff for player $i$ if all players play according to strategy profile $\langle \sigma_i, \sigma_{-i} \rangle$.

$\pi^{\sigma}(h) = \Pi_{h' \cdot a \sqsubseteq h} \sigma_{P(h')}(h',a)$ is the joint probability of reaching $h$ if all players play according to $\sigma$. $\pi^{\sigma}_i(h)$ is the contribution of player $i$ to this probability (that is, the probability of reaching $h$ if all players other than $i$, and chance, always chose actions leading to $h$). $\pi^{\sigma}_{-i}(h)$ is the contribution of chance and all players other than $i$.

A {\em best response} to $\sigma_{i}$ is a strategy $BR(\sigma_{i})$ such that $u_i\big(\sigma_i, BR(\sigma_{i})\big) = \max_{\sigma'_{-i}} u_i(\sigma_i, \sigma'_{-i})$.
A {\em Nash equilibrium} $\sigma^*$ is a strategy profile where everyone plays a best response: $\forall i$, $u_i(\sigma^*_i, \sigma^*_{-i}) = \max_{\sigma'_i} u_i(\sigma'_i, \sigma^*_{-i})$~\cite{Nash50:Eq}.
The {\em exploitability} $e(\sigma_i)$ of a strategy $\sigma_i$ in a two-player zero-sum game is how much worse it does versus a best response compared to a Nash equilibrium strategy. Formally, $e(\sigma_i) = u_i\big(\sigma_i^*,BR(\sigma_i^*)\big) - u_i\big(\sigma_i,BR(\sigma_i)\big)$. In an $\epsilon$-Nash equilibrium, no player has exploitability higher than $\epsilon$.

In CFR, the strategy vector for each infoset is determined according to a regret-minimization algorithm. Typically, {\em regret matching} (RM) is used as that algorithm within CFR due to RM's simplicity and lack of parameters.

The expected value (or simply \emph{value}) to player~$i$ at state $h$ given that all players play according to strategy profile $\sigma$ from that point on is defined as $v_i^{\sigma}(h)$. The value to $i$ at infoset $I$ where $i$ acts is the weighted average of the value of each state in the infoset, where the weight is proportional to $i$'s belief that they are in that state conditional on knowing they are in $I$. Formally, $v^{\sigma}(I) = \sum_{h \in I} \big(\pi^{\sigma}_{-i}(h|I)v_i^{\sigma}(h)\big)$ and $v^{\sigma}(I, a) = \sum_{h \in I} \big(\pi^{\sigma}_{-i}(h|I) v_i^{\sigma}(h \cdot a)\big)$ where $\pi^{\sigma}_{-i}(h|I) = \frac{\pi^{\sigma}_{-i}(h)}{\pi^{\sigma}_{-i}(I)}$.

Let $\sigma^t$ be the strategy on iteration $t$. The {\em instantaneous regret}
for action $a$ in infoset $I$ on iteration $t$ is
$r^t(I,a) = v^{\sigma^t}(I,a) - v^{\sigma^t}(I)$
and the {\em regret} on iteration $T$ is
\begin{equation}
R^T(I,a) = \sum_{t = 1}^T r^{T}(I,a)
\label{eq:regret}
\end{equation}
Additionally, $R^T_+(I,a) = \max\{R^T(I,a), 0 \}$ and $R^T(I) = \max_a\{R_+^T(I,a)\}$. Regret
for player $i$ in the entire game is
\begin{equation}
R_i^T = \max_{\sigma_i'} \sum_{t = 1}^T \big(u_i(\sigma'_i, \sigma_{-i}^t) - u_i(\sigma^t_i, \sigma_{-i}^t)\big)
\end{equation}

In RM, a player picks a distribution over actions in an infoset in proportion to the positive regret on those actions. Formally, on each iteration $T+1$, player $i$ selects actions $a \in A(I)$ according to probabilities
\begin{equation}
\sigma^{T+1}(I,a) =
\begin{cases}
\frac{R^T_+(I,a)}{\sum_{a' \in A(I)}R_+^T(I,a')}, & \text{if}\ \sum_{a'}R^T_+(I,a') > 0 \\
\frac{1}{|A(I)|}, & \text{otherwise}
\end{cases}
\label{eq:rm}
\end{equation}
If a player plays according to regret matching in infoset $I$ on every iteration, then on iteration $T$, $R^T(I) \le \Delta\sqrt{|A(I)|}\sqrt{T}$~\cite{Cesa-Bianchi06:Prediction}.

If a player plays according to CFR on every iteration, then
\begin{equation}
R_i^T \le \sum_{I \in \mathcal{I}_i} R^T(I)
\label{eq:bound}
\end{equation}
So, as $T \rightarrow \infty$, $\frac{R_i^T}{T} \rightarrow 0$.

The average strategy $\bar{\sigma}_i^T(I)$ for an infoset $I$ is
\begin{equation}
\bar{\sigma}_i^T(I) = \frac{\sum_{t = 1}^T \big(\pi_i^{\sigma^t}(I)\sigma_i^t(I)\big)}{\sum_{t = 1}^T \pi_i^{\sigma^t}(I)}
\label{eq:average}
\end{equation}

CFR minimizes external regret~\cite{Zinkevich07:Regret}, so it converges to a \emph{coarse correlated equilibrium}~\cite{Hart00:Simple}. In two-player zero-sum games, this is also a Nash equilibrium.
In two-player zero-sum games, if both players' average regret satisfies $\frac{R_i^T}{T} \le \epsilon$, then their average strategies $\langle \bar{\sigma}^T_1, \bar{\sigma}^T_2 \rangle$ are a $2\epsilon$-Nash equilibrium~\cite{Waugh09:Thesis}. Thus, CFR is an anytime algorithm for finding an $\epsilon$-Nash equilibrium in two-player zero-sum games. 

Although CFR theory calls for both players to simultaneously update their regrets on each iteration, in practice far better performance is achieved by alternating which player updates their regrets on each iteration. However, this complicates the theory for convergence~\cite{Farina19:Online,Burch18:Revisiting}.

CFR+ is like CFR but with the following small changes.
First, after each iteration any action with negative regret is set to zero regret. Formally, CFR+ chooses its strategy on iteration $T+1$ according to \emph{Regret Matching+ (RM+)}, which is identical to Equation~(\ref{eq:rm}) but uses the regret-like value $Q^T(I,a) = \max\{0, Q^{T-1}(I,a) + r^t(I,a)\}$ rather than $R^T_+(I,a)$. 
Second, CFR+ uses a weighted average strategy where iteration $T$ is weighted by $T$ rather than using a uniformly-weighted average strategy as in CFR. 
The best known convergence bound for CFR+ is higher (that is, worse in exploitability) than CFR by a constant factor of 2.
Despite that, CFR+ typically converges much faster than CFR and usually even faster than $O(\frac{1}{\epsilon})$.

However, in some games CFR+ converges slower than $\frac{1}{T}$. We now provide a two-player zero-sum game with this property. Consider the payoff matrix $\left[ \begin{smallmatrix} 1&0.9\\ -0.7&1 \end{smallmatrix} \right]$ (where $P_1$ chooses a row and $P_2$ simultaneously chooses a column; the chosen entry in the matrix is the payoff for $P_1$ while $P_2$ receives the opposite).
We now proceed to introducing our improvements to the CFR family.

\vspace{-0.00in}
\section{Weighted Averaging Schemes for CFR+}
\vspace{-0.02in}
As described in the previous section, CFR+ traditionally uses ``linear'' averaging, in which iteration $t$'s contribution to the average strategy is proportional to $t$. In this section we prove a bound for any sequence of non-decreasing weights when calculating the average strategy. However, the bound on convergence is never lower than that of vanilla CFR (that is, uniformly equal weight on the iterations).

\begin{theorem}
	\label{th:cfrp}
	Suppose $T$ iterations of RM+ are played in a two-player zero-sum game. Then the weighted average strategy profile, where iteration $t$ is weighed proportional to $w_t > 0$ and $w_i \le w_j$ for all $i < j$, is a $\frac{w_T}{\sum_{t = 1}^T w_t} \Delta|\mathcal{I}|\sqrt{|A|}\sqrt{T}$-Nash equilibrium.
\end{theorem}

The proof is in the appendix. It largely follows the proof for linear averaging in CFR+~\cite{Tammelin15:Solving}.

Empirically we observed that CFR+ converges faster when assigning iteration $t$ a weight of $t^2$ rather than a weight of $t$ when calculating the average strategy. We therefore use this weight for CFR+ and its variants throughout this paper when calculating the average strategy.

\vspace{-0.02in}
\section{Regret Discounting for CFR and Its Variants}
\vspace{-0.02in}
\label{sec:discount}

In all past variants of CFR, each iteration's contribution to the \emph{regrets} is assigned equal weight. In this section we discuss discounting iterations in CFR when determining regrets---in particular, assigning less weight to earlier iterations. This is very different from, and orthogonal to, the idea of discounting iterations when computing the average strategy, described in the previous section. 

To motivate discounting, consider the simple case of an agent deciding between three actions. The payoffs for the actions are 0, 1, and -1,000,000, respectively. From (\ref{eq:rm}) we see that CFR and CFR+ assign equal probability to each action on the first iteration. This results in regrets of 333,333, 333,334, and 0, respectively. If we continue to run CFR or CFR+, the next iteration will choose the first and second action with roughly 50\% probability each, and the regrets will be updated to be roughly 333,332.5 and 333,334.5, respectively. It will take 471,407 iterations for the agent to choose the second action---that is, the best action---with 100\% probability. Discounting the first iteration over time would dramatically speed convergence in this case. While this might seem like a contrived example, many games include highly suboptimal actions. In this simple example the bad action was chosen on the first iteration, but in general bad actions may be chosen throughout a run, and discounting may be useful far beyond the first few iterations.

Discounting prior iterations has received relatively little attention in the equilibrium-finding community. ``Optimistic'' regret minimizing variants exist that assign a higher weight to recent iterations, but this extra weight is temporary and typically only applies to a short window of recent iterations; for example, counting the most recent iterate twice~\cite{Syrgkanis15:Fast}. We investigate optimistic regret minimizers as part of CFR later in this paper. CFR+ discounts prior iterations' contribution to the \emph{average strategy}, but not the \emph{regrets}. Discounting prior iterations has also been used in CFR for situations where the game structure changes, for example due to interleaved abstraction and equilibrium finding~\cite{Brown14:Regret,Brown15:Simultaneous}. There has also been some work on applying discounting to perfect-information game solving in Monte Carlo Tree Search~\cite{Hashimoto11:Accelerated}.

Outside of equilibrium finding, prior research has analyzed the theory for discounted regret minimization~\cite{Cesa-Bianchi06:Prediction}. That work investigates applying RM (and other regret minimizers) to a sequence of iterations in which iteration $t$ has weight $w_t$ (assuming $w_t \le 1$ and the final iteration has weight $1$). For RM, it proves that if $\sum_{t=1}^{\infty} w_t = \infty$ then weighted average regret, defined as $R_i^{w,T} = \max_{a \in A} \frac{\sum_{t=1}^T(w_t r^t(a))}{\sum_{t=1}^T w^t}$ is bounded by
\begin{equation}
\label{eq:discount_bound}
R_i^{w,T} \le \frac{\Delta \sqrt{|A|}\sqrt{\sum_{t=1}^Tw_t^2}}{\sum_{t=1}^Tw_t}
\end{equation}
Prior work has shown that, in two-player zero-sum games, if weighted average regret is $\epsilon$, then the weighted average strategy, defined as $\sigma_i^{w,T}(I) = \frac{\sum_{t \in T} \big(w_t \pi_i^{\sigma^t}(I)\sigma_i^t(I)\big)}{\sum_{t \in T} (w_t \pi_i^{\sigma^t}(I))}$ for infoset $I$, is a $2\epsilon$-Nash equilibrium~\cite{Brown14:Regret}.

While there are a limitless number of discounting schemes that converge in theory, not all of them perform well in practice. This paper introduces a number of variants that perform particularly well also in practice. The first algorithm, which we refer to as \emph{linear CFR (LCFR)}, is identical to CFR, except on iteration $t$ the updates to the regrets and average strategies are given weight $t$. That is, the iterates are weighed linearly. (Equivalently, one could multiply the accumulated regret by $\frac{t}{t + 1}$ on each iteration. We do this in our experiments to reduce the risk of numerical instability.) This means that after $T$ iterations of LCFR, the first iteration only has a weight of $\frac{2}{T^2 + T}$ on the regrets rather than a weight of $\frac{1}{T}$, which would be the case in CFR and CFR+. In the motivating example introduced at the beginning of this section, LCFR chooses the second action with 100\% probability after only 970 iterations while CFR+ requires 471,407 iterations. Furthermore, from (\ref{eq:discount_bound}), the theoretical bound on the convergence of regret is only greater than vanilla CFR by a factor of $\frac{2}{\sqrt{3}}$. One could more generally use any polynomial weighting of $t$.

Since the changes from CFR that lead to LCFR and CFR+ do not conflict, it is natural to attempt to combine them into a single algorithm that weighs each iteration $t$ proportional to $t$ and also has a floor on regret at zero like CFR+. However, we empirically observe that this algorithm, which we refer to as \emph{LCFR+}, actually leads to performance that is \emph{worse} than LCFR and CFR+ in the games we tested, even though its theoretical bound on convergence is the same as for LCFR.

Nevertheless, we find that using a less-aggressive discounting scheme leads to consistently strong performance. We can consider a family of algorithms called \emph{Discounted CFR} with parameters $\alpha$ $\beta$, and $\gamma$ \emph{(DCFR$_{\alpha,\beta,\gamma}$)}, defined by multiplying accumulated positive regrets by $\frac{t^\alpha}{t^\alpha + 1}$, negative regrets by $\frac{t^\beta}{t^\beta + 1}$, and contributions to the average strategy by $(\frac{t}{t + 1})^\gamma$ on each iteration $t$. In this case, LCFR is equivalent to DCFR$_{1,1,1}$, because multiplying iteration $t$'s regret and contribution to the average strategy by $\frac{t'}{t'+1}$ on every iteration $t \le t' < T$ is equivalent to weighing iteration $t$ by $\frac{t}{T}$. CFR+ (where iteration $t$'s contribution to the average strategy is proportional to $t^2$) is equivalent to DCFR$_{\infty,-\infty,2}$.

In preliminary experiments we found the optimal choice of $\alpha$, $\beta$, and $\gamma$ varied depending on the specific game. However, we found that setting $\alpha = 3/2$, $\beta = 0$, and $\gamma = 2$ led to performance that was consistently stronger than CFR+. Thus, when we refer to DCFR with no parameters listed, we assume this set of parameters are used.

Theorem~\ref{th:dcfr} shows that DCFR has a convergence bound that differs from CFR only by a constant factor.

\begin{theorem}
	\label{th:dcfr}
	Assume that $T$ iterations of DCFR are conducted in a two-player zero-sum game. Then the weighted average strategy profile is a $6\Delta|\mathcal{I}|(|\sqrt{|A|} + \frac{1}{\sqrt{T}})/\sqrt{T}$-Nash equilibrium.
\end{theorem}

We provide the proof in the appendix. It combines elements of the proof for CFR+~\cite{Tammelin15:Solving} and the proof that discounting in regret minimization is sound~\cite{Cesa-Bianchi06:Prediction}.

One of the drawbacks of setting $\beta \le 0$ is that suboptimal actions (that is, actions that have an expected value lower than some other action in every equilibrium) no longer have regrets that approach $-\infty$ over time. Instead, for $\beta = 0$ they will approach some constant value and for $\beta < 0$ they will approach $0$. This makes the algorithm less compatible with improvements that prune negative-regret actions~\cite{Brown15:Regret-Based,Brown17:Reduced}. Such pruning algorithms can lead to more than an order of magnitude reduction in computational and space requirements for some games. Setting $\beta > 0$ better facilitates this pruning. For this reason in our experiments we also show results for $\beta = 0.5$.

\vspace{-0.02in}
\section{Experimental setup}
\vspace{-0.01in}
We now introduce the games used in our experiments.
\vspace{-0.02in}
\subsection{Description of heads-up no-limit Texas hold'em}
\vspace{-0.04in}
We conduct experiments on subgames of HUNL poker, a primary benchmark for imperfect-information game solving. In the version of HUNL we use, and which is standard in the Annual Computer Poker Competition, the two players ($P_1$ and $P_2$) start each hand with \$20,000. The players alternate positions after each hand. On each of the four rounds of betting, each player can choose to either fold, call, or raise. Folding results in the player losing and the money in the pot being awarded to the other player. Calling means the player places a number of chips in the pot equal to the opponent's share. Raising means the player adds more chips to the pot than the opponent's share. A round ends when a player calls (if both players have acted). Players cannot raise beyond the \$20,000 they start with. All raises must be at least \$100 and at least as larger as any previous raise on that round.

At the start of each hand of HUNL, both players are dealt two private cards from a standard 52-card deck. $P_1$ places \$100 in the pot and $P_2$ places \$50 in the pot. A round of betting then occurs. Next, three \emph{community} cards are dealt face up. Another round of betting occurs, starting with $P_1$. After the round is over, another community card is dealt face up, and another round of betting starts with $P_1$ acting first. Finally, one more community card is revealed and a final betting round occurs starting with $P_1$. Unless a player has folded, the player with the best five-card poker hand, constructed from their two private cards and the five community cards, wins the pot. In the case of a tie, the pot is split evenly.

Although the HUNL game tree is too large to traverse completely without sampling, state-of-the-art agents for HUNL solve subgames of the full game in real time during play~\cite{Brown17:Safe,Moravcik17:DeepStack,Brown17:Superhuman,Brown18:Depth} using a small number of the available bet sizes. For example, \emph{Libratus} solved in real time the remainder of HUNL starting on the third betting round. We conduct our HUNL experiments on four subgames generated by \emph{Libratus}~\footnote{https://github.com/CMU-EM/LibratusEndgames}. The subgames were selected prior to testing. Although the inputs to the subgame are publicly available (the beliefs of both players at the start of the subgame about what state they are in, the number of chips in the pot, and the revealed cards), the exact bet sizes that \emph{Libratus} considered have not been publicly revealed. We therefore use the bet sizes of 0.5x and 1x the size of the pot, as well as an all-in bet (betting all remaining chips) for the first bet of each round. For subsequent bets in a round, we consider 1x the pot and all-in.

Subgame 1 begins at the start of the third betting round and continues to the end of the game. There are \$500 in the pot at the start of the round. This is the most common situation to be in upon reaching the third betting round, and is also the hardest for AIs to solve because the remaining game tree is the largest. Since there is only \$500 in the pot but up to \$20,000 could be lost, this subgames contains a number of high-penalty mistake actions. Subgame 2 begins at the start of the third betting round and has \$4,780 in the pot at the start of the round.
Subgame 3 begins at the start of the fourth (and final) betting round with \$500 in the pot, which is a common situation. Subgame 4 begins at the start of the fourth betting round with \$3,750 in the pot. Exploitability is measured in terms of milli big blinds per game (mbb/g), a standard measurement in the field, which represents the number of big blinds ($P_1$'s original contribution to the pot) lost per hand of poker multiplied by 1,000.
\vspace{-0.02in}
\subsection{Description of Goofspiel}
\vspace{-0.04in}
In addition to HUNL subgames, we also consider a version of the game of Goofspiel (limited to just five cards per player). In this version of Goofspiel, each player has five hidden cards in their hand (A, 2, 3, 4, and 5), with A being valued as 1. A deck of five cards (also of rank A, 2, 3, 4, and 5), is placed between the two players. In the variant we consider, both players know the order of revealed cards in the center will be A, 2, 3, 4, 5. On each round, the top card of the deck is flipped and is considered the prize card. Each player then simultaneously plays a card from their hand. The player who played the higher-ranked card wins the prize card. If the players played the same rank, then they split the prize's value. The cards that were bid are discarded. At the end of the game, players add up the ranks of their prize cards. A player's payoff is the difference between his total value and the total value of his opponent.

\vspace{-0.02in}
\section{Experiments on Regret Discounting and Weighted Averaging}
\vspace{-0.01in}
Our experiments are run for 32,768 iterations for HUNL subgames and 8,192 iterations for Goofspiel. Since all the algorithms tested only converge to an $\epsilon$-equilibrium rather than calculating an exact equilibrium, it is up to the user to decide when a solution is sufficiently converged to terminate a run. In practice, this is usually after 100 - 1,000 iterations~\cite{Brown17:Superhuman,Moravcik17:DeepStack}. For example, an exploitability of 1 mbb/g is considered sufficiently converged so as to be essentially solved~\cite{Bowling15:Heads-up}. Thus, the performance of the presented algorithms between 100 and 1,000 iterations is arguably more important than the performance beyond 10,000 iterations. Nevertheless, we show performance over a long time horizon to display the long-term behavior of the algorithms. All our experiments use the alternating-updates form of CFR. We measure the average exploitability of the two players.

Our experiments show that LCFR can dramatically improve performance over CFR+ over reasonable time horizons in certain games. However, asymptotically, LCFR appears to do worse in practice than CFR+. LCFR does particularly well in subgame 1 and 3, which (due to the small size of the pot relative to the amount of money each player can bet) have more severe mistake actions compared to subgames 2 and 4. It also does poorly in Goofspiel, which also likely does not have severely suboptimal actions. This suggests that LCFR is particularly well suited for games with the potential for large mistakes.

Our experiments also show that DCFR$_{\frac{3}{2},0,2}$ matches or outperforms CFR+ across the board. The improvement is usually a factor of 2 or 3. In Goofspiel, DCFR$_{\frac{3}{2},0,2}$ results in essentially identical performance as CFR+.

DCFR$_{\frac{3}{2},-\infty,2}$, which sets negative regrets to zero rather than multiplying them by $\frac{1}{2}$ each iteration, generally also leads to equally strong performance, but in rare cases (such as in Figure~\ref{fig:subgame2}), can produce a spike in exploitability that takes many iterations to recover from. Thus, we generally recommend using DCFR$_{\frac{3}{2},0,2}$ over DCFR$_{\frac{3}{2},-\infty,2}$.

DCFR$_{\frac{3}{2},\frac{1}{2},2}$ multiplies negative regrets by $\frac{\sqrt{t}}{\sqrt{t} + 1}$ on iteration $t$, which allows suboptimal actions to decrease in regret to $-\infty$ and thereby facilitates algorithms that temporarily prune negative-regret sequences. In the HUNL subgames, DCFR$_{\frac{3}{2},\frac{1}{2},2}$ performed very similarly to DCFR$_{\frac{3}{2},0,2}$. However, in Goofspiel it does noticeably worse. This suggests that DCFR$_{\frac{3}{2},\frac{1}{2},2}$ may be preferable to DCFR$_{\frac{3}{2},0,2}$ in games with large mistakes when a pruning algorithm may be used, but that DCFR$_{\frac{3}{2},0,2}$ should be used otherwise.

\begin{figure}[!h]
	\centering
	\includegraphics[width=83mm]{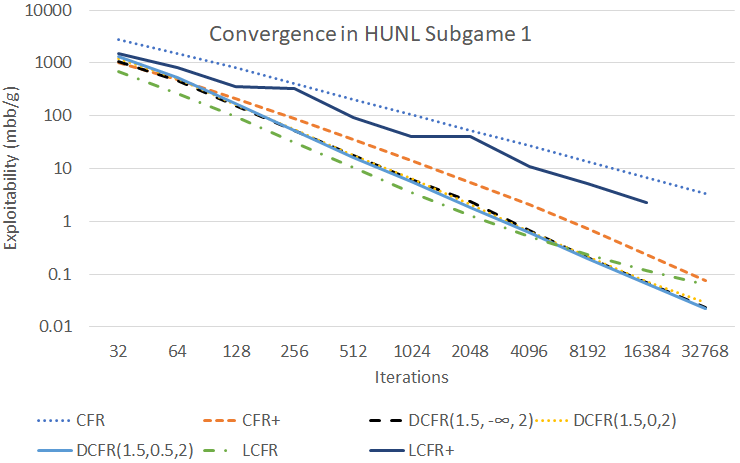}
	\caption{Convergence in HUNL Subgame1.}
	\label{fig:subgame1}
	\vspace{-0.05in}
\end{figure}

\begin{figure}[!h]
	\centering
	\includegraphics[width=83mm]{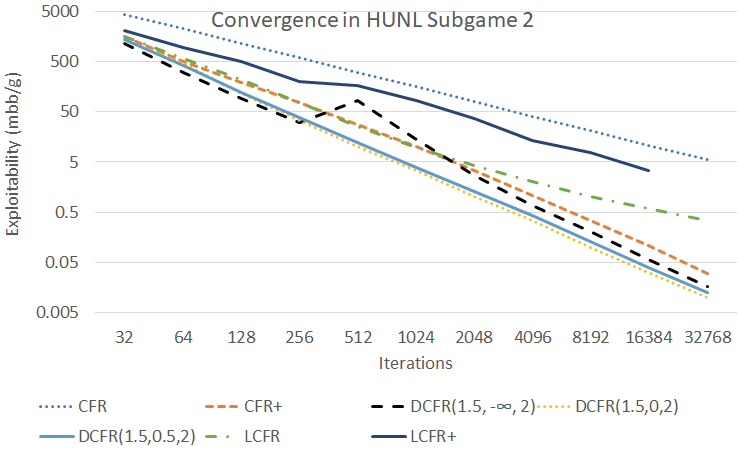}
	\caption{Convergence in HUNL Subgame2.}
	\label{fig:subgame2}
	\vspace{-0.05in}
\end{figure}

\begin{figure}[!h]
	\centering
	\includegraphics[width=83mm]{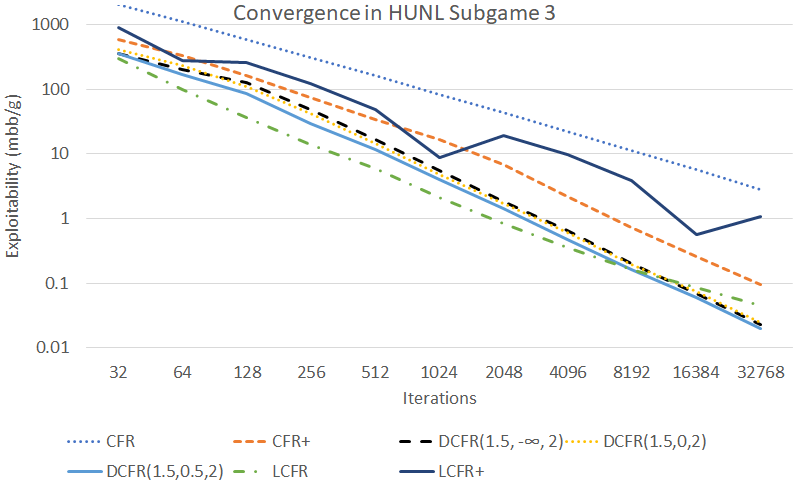}
	\caption{Convergence in HUNL Subgame 3.}
	\label{fig:subgame3}
	\vspace{-0.05in}
\end{figure}

\begin{figure}[!h]
	\centering
	\includegraphics[width=83mm]{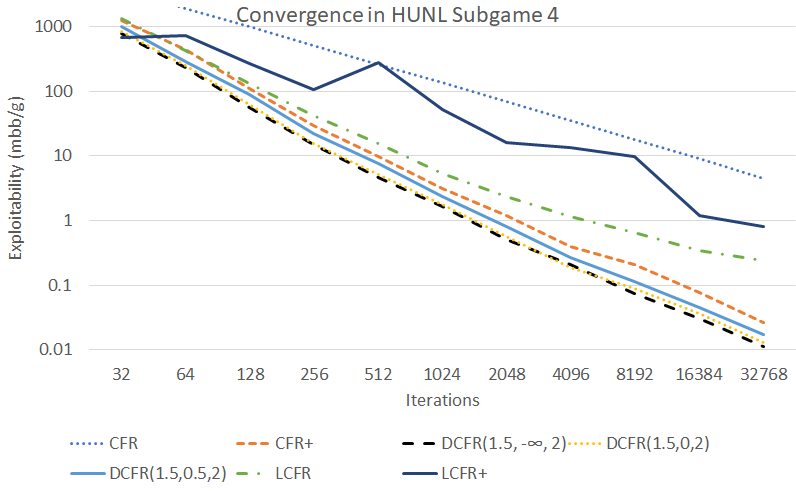}
	\caption{Convergence in HUNL Subgame 4.}
	\label{fig:subgame4}
	\vspace{-0.05in}
\end{figure}

\begin{figure}[!h]
	\centering
	\includegraphics[width=83mm]{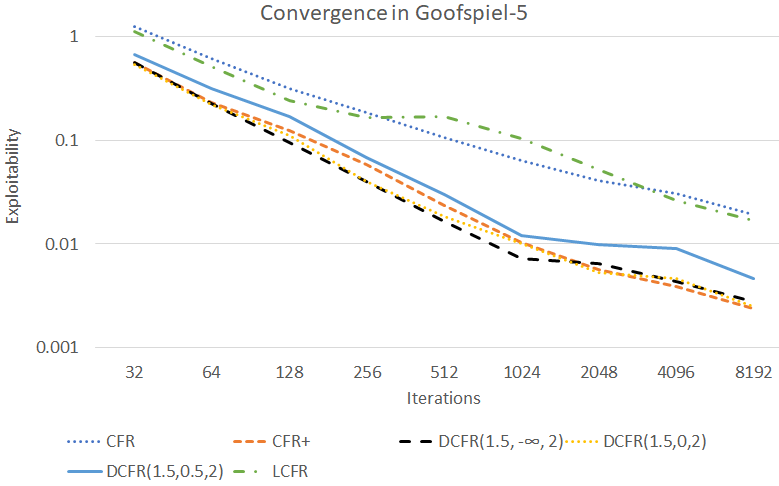}
	\caption{Convergence in 5-card Goofspiel variant.}
	\label{fig:goofspiel}
	\vspace{-0.05in}
\end{figure}
\vspace{-0.02in}
\section{NormalHedge for CFR Variants}
\vspace{-0.01in}
CFR is a framework for applying regret minimization independently at each infoset in the game. Typically RM is used as the regret minimizer primarily due to its lack of parameters and its simple implementation. However, any regret minimizer can be applied. Previous research investigated using Hedge~\cite{Littlestone94:Weighted,Freund97:Decision-Theoretic} in CFR rather than RM~\cite{Brown17:Dynamic}. This led to better performance in small games, but worse performance in large games. In this section we investigate instead using NormalHedge (NH)~\cite{Chaudhuri09:Parameter-free} as the regret minimizer in CFR.

In NH, on each iteration $T+1$ a player $i$ selects actions $a \in A(I)$ proportional to $\frac{R_+^T(I,a)}{c_t}\exp\big(\frac{(R_+^T(I,a))^2}{2c_t}\big)$ where $c_t > 0$ satisfies $\frac{1}{N}\sum_{i=1}^N\exp\big(\frac{(R_+^T(I,a))^2}{2c_t}\big) = e$. If a player plays according to NH in infoset $I$, then cumulative regret for that infoset is at most $O(\Delta\sqrt{T\ln(|A|)} + \Delta\ln^2(|A|))$. %This is better than the $O(\Delta\sqrt{T|A|})$ bound for RM.

NH shares two desirable properties with RM: it does not have any parameters and it assigns zero probability to actions with negative regret (which means it can be easily used in CFR+ with a floor on regret at zero). However, the NH operation is more computationally expensive than RM because it involves exponentiation and a line search for $c_t$.

In our experiments we investigate using NH in place of RM for DCFR$_{\frac{3}{2},0,2}$ and in place of RM for LCFR. We found that NH did worse in all HUNL subgames compared to RM in LCFR, so we omit those results. Figure~\ref{fig:nh1} and Figure~\ref{fig:nh3} shows that NH outperforms RM in HUNL subgames when combined with  DCFR$_{\frac{3}{2},0,2}$. However, it does worse than RM in Figure~\ref{fig:nh2} and Figure~\ref{fig:nh4}. The two subgames it does better in have the largest ``mistake'' actions, which suggest NH may do better in games that have large mistake actions.

In these experiments the performance of NH is measured in terms of exploitability as a function of number of iterations. However, in our implementation, each iteration takes five times longer due to the exponentiation and line search operations involved in NH. Thus, using NH actually slows convergence in practice. Nevertheless, NH may be preferable in certain situations where the cost of the exponentiation and line search operations are insignificant, such as when an algorithm is bottlenecked by memory access rather than computational speed.
\vspace{-0.02in}
\section{Optimistic CFR Variants}
\vspace{-0.01in}
Optimistic Hedge~\cite{Syrgkanis15:Fast} is a regret minimization algorithm similar to Hedge in which the last iteration is counted twice when determining the strategy for the next iteration. This can lead to substantially faster convergence, including in some cases an improvement over the $O(\frac{1}{\epsilon^2})$ bound on regret of typical regret minimizers.

We investigate counting the last iteration twice when calculating the strategy for the next iteration~\cite{Burch17:Time}. Formally, when applying Equation~(\ref{eq:rm}) to determine the strategy for the next iteration, we use a modified regret $R_{\textit{mod}}^T(I,a) = \sum_{t=1}^{T-1}r^t(I,a) + 2r^T(I,a)$ in the equation in place of $R^T(I,a)$. We refer to this as Optimistic RM, and any CFR variant that uses it as Optimistic. We found that Optimistic DCFR$_{\frac{3}{2},0,2}$ did worse than DCFR$_{\frac{3}{2},0,2}$ in all HUNL subgames, so we omit those results. Figure~\ref{fig:nh1} and Figure~\ref{fig:nh3} shows that Optimistic LCFR outperforms LCFR in two HUNL subgames. However, it does worse than LCFR in Figure~\ref{fig:nh2} and Figure~\ref{fig:nh4}. Just as in the case of NH, the two subgames that Optimistic LCFR does better in have the largest ``mistake'' actions, which suggests that Optimistic LCFR may do better than LCFR in games that have large mistake actions. These are the same situations that LCFR normally excels in, so this suggests that in a situation where LCFR is preferable, one may wish to use Optimistic LCFR.

\begin{figure}[!h]
	\centering
	\includegraphics[width=83mm]{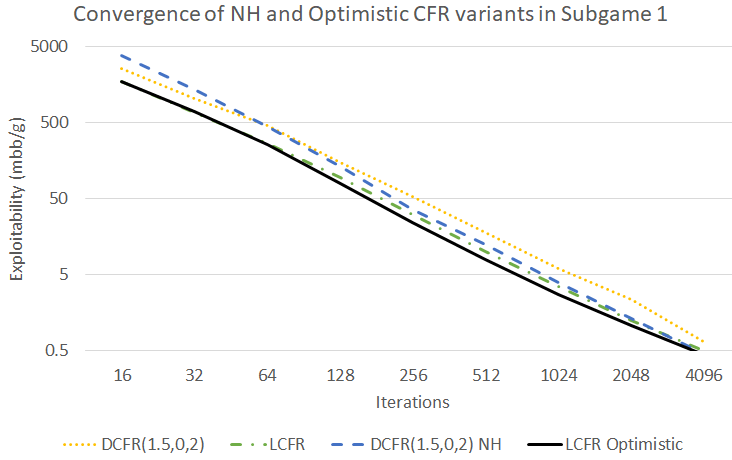}
	\caption{Convergence in HUNL Subgame 1.}
	\label{fig:nh1}
	\vspace{-0.05in}
\end{figure}

\begin{figure}[!h]
	\centering
	\includegraphics[width=83mm]{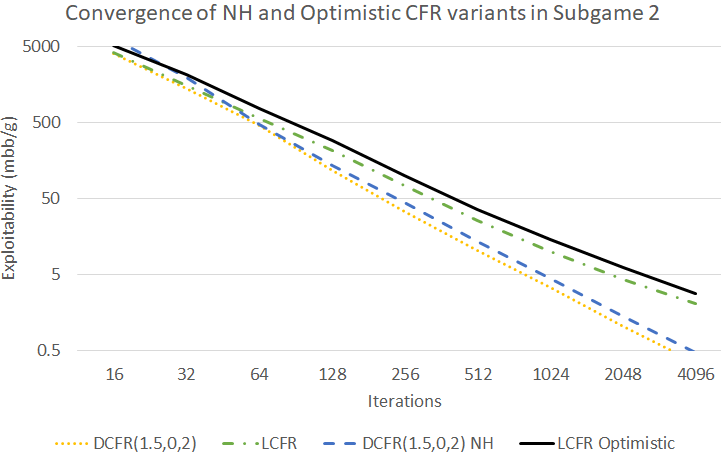}
	\caption{Convergence in HUNL Subgame 2.}
	\label{fig:nh2}
	\vspace{-0.05in}
\end{figure}

\begin{figure}[!h]
	\centering
	\includegraphics[width=83mm]{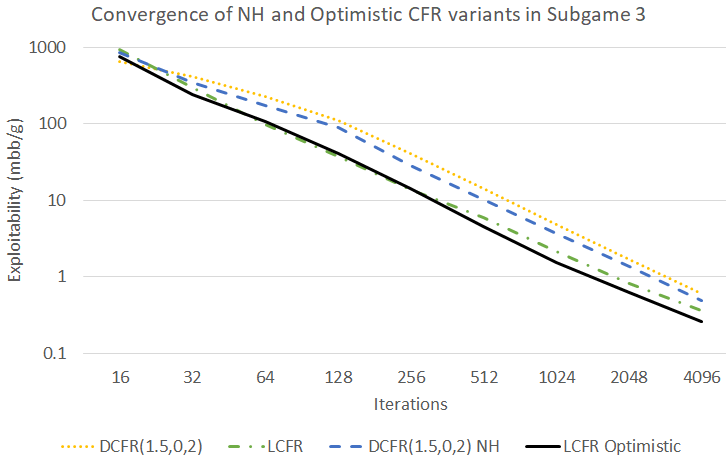}
	\caption{Convergence in HUNL Subgame 3.}
	\label{fig:nh3}
	\vspace{-0.05in}
\end{figure}

\begin{figure}[!h]
	\centering
	\includegraphics[width=83mm]{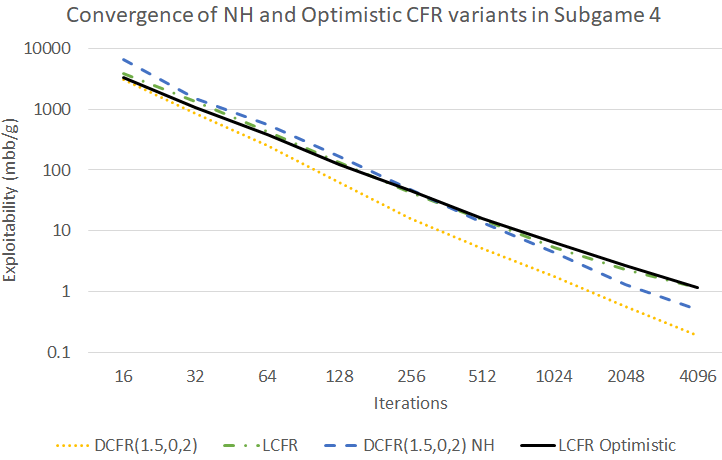}
	\caption{Convergence in HUNL Subgame 4.}
	\label{fig:nh4}
	\vspace{-0.05in}
\end{figure}

\vspace{-0.01in}
\section{Discounted Monte Carlo CFR}
\vspace{-0.01in}
\emph{Monte Carlo CFR (MCCFR)} is a variant of CFR in which certain player actions or chance outcomes are sampled~\cite{Lanctot09:Monte,Gibson12:Generalized}. MCCFR combined with abstraction has produced state-of-the-art HUNL poker AIs~\cite{Brown17:Superhuman}. It is also particularly useful in games that do not have a special structure that can be exploited to implement a fast vector-based implementation of CFR~\cite{Lanctot09:Monte,Johanson11:Accelerating}. There are many forms of MCCFR with different sampling schemes. The most popular is external-sampling MCCFR, in which opponent and chance actions are sampled according to their probabilities, but all actions belonging to the player updating his regret are traversed. Other MCCFR variants exist that achieve superior performance~\cite{Jackson17:Targeted}, but external-sampling MCCFR is simple and widely used, which makes it useful as a benchmark for our experiments.

\begin{figure}[!h]
	\centering
	\includegraphics[width=83mm]{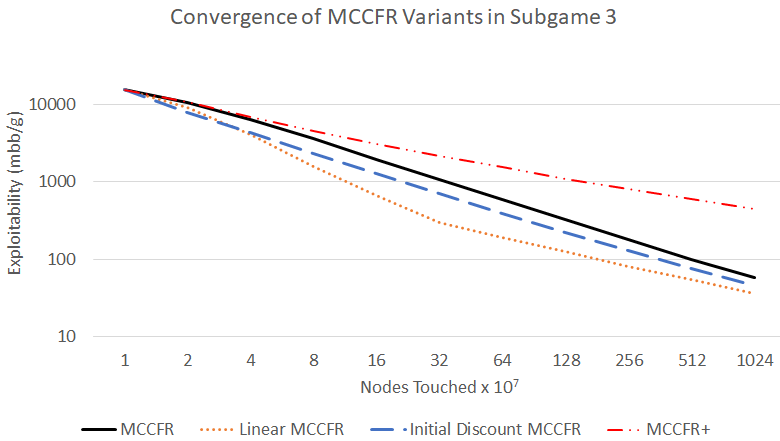}
	\caption{Convergence of MCCFR in HUNL Subgame 3.}
	\label{fig:mccfr3}
	\vspace{-0.05in}
\end{figure}

\begin{figure}[!h]
	\centering
	\includegraphics[width=83mm]{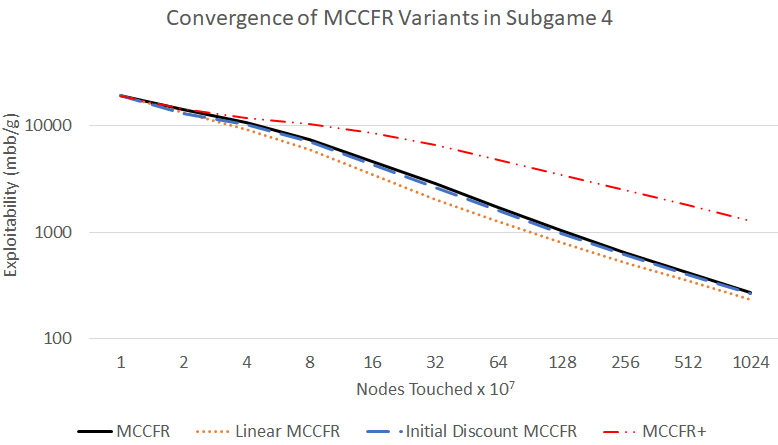}
	\caption{Convergence of MCCFR in HUNL Subgame 4.}
	\label{fig:mccfr4}
	\vspace{-0.05in}
\end{figure}

Although CFR+ provides a massive improvement over CFR in the unsampled case, the changes present in CFR+ (a floor on regret at zero and linear averaging), do not lead to superior performance when applied to MCCFR~\cite{Burch17:Time}. In contrast, in this section we show that the changes present in LCFR do lead to superior performance when applied to MCCFR. Specifically, we divide the MCCFR run into periods of $10^7$ nodes touched. Nodes touched is an implementation-independent and hardware-independent proxy for time that counts the number of nodes traversed (including terminal nodes). After each period $n$ ends, we multiply all accumulated regrets and contributions to the average strategies by $\frac{n}{n+1}$. Figure~\ref{fig:mccfr3} and Figure~\ref{fig:mccfr4} demonstrate that this leads to superior performance in HUNL compared to vanilla MCCFR. The improvement is particularly noticeable in subgame 3, which features the largest mistake actions. We also show performance if one simply multiplies the accumulated regrets and contributions to the average strategy by $\frac{1}{10}$ after the first period ends, and thereafter runs vanilla MCCFR (the ``Initial Discount MCCFR'' variant). The displayed results are the average of 100 different runs.

\vspace{-0.01in}
\section{Conclusions}
\vspace{-0.01in}
We introduced variants of CFR that discount prior iterations, leading to stronger performance than the prior state-of-the-art CFR+, particularly in settings that involve large mistakes. In particular, the DCFR$_{\frac{3}{2},0,2}$ variant matched or outperformed CFR+ in all settings.

\section{Acknowledgments}

This material is based on work supported by the National Science Foundation under grants IIS-1718457, IIS-1617590, and CCF-1733556, and the ARO under award W911NF-17-1-0082. Noam is also sponsored by an Open Philanthropy Project AI Fellowship and a Tencent AI Lab Fellowship.

\bibliographystyle{aaai}

\clearpage

\newpage

\appendix
\noindent {\LARGE \bf Appendix}

\section{Proof of Theorem~\ref{th:cfrp}}

Consider the weighted sequence of iterates $\sigma'^1,...,\sigma'^T$ in which $\sigma'^t$ is identical to $\sigma^t$, but weighed by $w_t$. The regret of action $a$ in infoset $I$ on iteration $t$ of this new sequence is $R'^t(I,a)$.

From Lemma~\ref{le:q} we know that $R^t(I,a) \le \Delta\sqrt{|A|}\sqrt{T}$ for player~$i$ for action $a$ in infoset $I$. Since $w_{a,t}$ is a non-decreasing sequence, so we can apply Lemma~\ref{le:plausible} using weight $w_{t}$ for iteration $t$ with $B = \Delta\sqrt{|A|}\sqrt{T}$ and $C = 0$. From Lemma~\ref{le:plausible}, this means that $R'^t(I,a) \le w_T \Delta \sqrt{|A|} \sqrt{T}$. Applying (\ref{eq:bound}), we get weighted regret is at most $w_T \Delta|\mathcal{I}_i|\sqrt{|A|}{\sqrt{T}}$ for player~$i$. Thus, weighted average regret is at most $\frac{w_T \Delta|\mathcal{I}_i|\sqrt{|A|}{\sqrt{T}}}{\sum_{t=1}^T w_t}$ . Since $|\mathcal{I}_1| + |\mathcal{I}_2| = |\mathcal{I}|$, so the weighted average strategies form a $\frac{w_T \Delta|\mathcal{I}|\sqrt{|A|}{\sqrt{T}}}{\sum_{t=1}^T w_t}$-Nash equilibrium.
	
\section{Proof Theorem~\ref{th:dcfr}}

\begin{proof}
	
	Since the lowest amount of instantaneous regret on any iteration is $-\Delta$ and DCFR multiplies negative regrets by $\frac{1}{2}$ each iteration, so regret for any action at any point is greater than $-2\Delta$.
	
	Consider the weighted sequence of iterates $\sigma'^1,...,\sigma'^T$ in which $\sigma'^t$ is identical to $\sigma^t$, but weighed by $w_{a,t} = \Pi_{i = t}^{T-1} \frac{i^2}{(i + 1)^2} = \frac{6t^2}{T(T+1)(2T+1)}$ rather than $w_t = \Pi_{i = t}^{T-1} \frac{i^{3/2}}{i^{3/2} + 1}$. The regret of action $a$ in infoset $I$ on iteration $t$ of this new sequence is $R'^t(I,a)$.
	
	From Lemma~\ref{le:dcfr_regret} we know that $R^t(I,a) \le 2\Delta\sqrt{|A|}\sqrt{T}$ for player~$i$ for action $a$ in infoset $I$. Since $w_{a,t}$ is an increasing sequence, so we can apply Lemma~\ref{le:plausible} using weight $w_{a,t}$ for iteration $t$ with $B = 2\Delta\sqrt{|A|}\sqrt{T}$ and $C = -2\Delta$. From Lemma~\ref{le:plausible}, this means that $R'^t(I,a) \le \frac{6 T^2 (2\Delta \sqrt{|A|} \sqrt{T} + 2\Delta)}{\big(T(T+1)(2T + 1)\big)} \le 6\Delta(|\sqrt{|A|} + \frac{1}{\sqrt{T}})/\sqrt{T}$. Applying (\ref{eq:bound}), we get weighted regret is at most $6\Delta|\mathcal{I}_i|(|\sqrt{|A|} + \frac{1}{\sqrt{T}})/\sqrt{T}$. Since the weights sum to one, this is also weighted average regret. Since $|\mathcal{I}_1| + |\mathcal{I}_2| = |\mathcal{I}|$, so the weighted average strategies form a $6\Delta(|\mathcal{I}|(\sqrt{|A|} + \frac{1}{\sqrt{T}})/\sqrt{T}$-Nash equilibrium.
\end{proof}

\begin{lemma}
	\label{le:plausible}
	Call a sequence $x_1,...,x_T$ of bounded real values $BC$-plausible if $B > 0$, $C \le 0$, $\sum_{t = 1}^i x_t \ge C$ for all $i$, and $\sum_{t=1}^T x_t \le B$. For any $BC$-plausible sequence and any sequence of non-decreasing weights $w_t \ge 0$, $\sum_{t = 1}^T (w_t x_t) \le w_T (B - C)$.
\end{lemma}

\begin{proof}
	The lemma closely resembles Lemma~3 from \cite{Tammelin15:Solving} and the proof shares some elements.
	
	We construct a $BC$-plausible sequence $x^*_1,...,x^*_T$ that maximizes the weighted sum. That is, $\sum_{t=1}^T w_t x'_t = \max_{x'_1,...,x'_T}\sum_{t=1}^T w_t x'_t$. We show that $x^*_1 = C$, $x^*_t = 0$ for $1 < t < T$, and $x^*_T = (B - C)$.
	
	Consider $x^*_T$. Clearly in order to maximize the weighted sum, $x^*_T = B - \sum_{t=1}^{T-1} (w_t x^*_t)$. Next, consider $x^*_t$ for $t < T$ and assume $x^*_{t'} = C - \sum_{t=1}^{t'} (w_t x^*_t)$ for $t < t' < T$ and assume $x^*_T = B - \sum_{t=1}^{T-1} (w_t x^*_t)$. Since $w_t \le w_T$ and $w_t \le w_{t'}$, so $\sum_{i=t}^{T} (w_i x^*_i)$ would be maximized if $x^*_t = C - \sum_{i=1}^{t-1} (w_i x^*_i)$. By induction, this means $x^*_1 = C$, $x^*_t = 0$ for $1 < t < T$, and $x^*_T = B - C$. In this case $\sum_{t=1}^T (w_t x^*_t) \le w_T (B-C) + w_1 C \le w_T (B - C)$. Since $x^*$ is a maximizing sequence, so for any sequence $x$ we have that $\sum_{t = 1}^T (w_t x_t) \le w_T (B - C)$.
\end{proof}

\begin{lemma}
	\label{le:q}
Given a sequence of strategies $\sigma^1,...,\sigma^T$, each defining a probability distribution over a set of actions $A$, consider any definition for $Q^t(a)$ satisfying the following conditions:
\begin{enumerate}
\item $Q^0(a) = 0$
\item $Q^t(a) = Q^{t-1}(a) + r^t(a)$ if $Q^{t-1}(a) + r^t(a) > 0$
\item $0 \ge Q^t(a) \ge Q^{t-1}(a) + r^t(a)$ if $Q^{t-1}(a) + r^t(a) \le 0$
\end{enumerate}
The regret-like value $Q^t(a)$ is then an upper bound on the regret $R^t(a)$ and $Q^t(a) - Q^{t-1}(a) \ge r^t(a) = R^t(a) - R^{t-1}(a)$.
\end{lemma}

\begin{proof}
	The lemma and proof closely resemble Lemma~1 in \cite{Tammelin15:Solving}. For any $t \ge 1$ we have $Q^{t+1}(a) - Q^t(a) \ge Q^t(a) + r^{t+1}(a) - Q^t(a) = R^{t+1}(a) - R^t(a)$. Since $Q^0(a) = 0$ and $R^0(a) = 0$, so $Q^t(a) \ge R^t(a)$.
\end{proof}

\begin{lemma}
	\label{le:neg}
	Given a set of actions $A$ and any sequence of rewards $v^t$ such that $|v^t(a) - v^t(b)| \le \Delta$ for all $t$ and all $a,b \in A$, after playing a sequence of strategies determined by regret matching but using the regret-like value $Q^t(a)$ in place of $R^t(a)$, $Q^T(a) \le \Delta \sqrt{|A|T}$ for all $a \in A$.
\end{lemma}

\begin{proof}
	The proof is identical to that of Lemma~2 in \cite{Tammelin15:Solving}.
\end{proof}

\begin{lemma}
	\label{le:dcfr_regret}
	Assume that player~$i$ conducts $T$ iterations of DCFR. Then weighted regret for the player is at most $\Delta |\mathcal{I}_i| \sqrt{|A|} \sqrt{T}$ and weighted average regret for the player is at most $2 \Delta |\mathcal{I}_i| \sqrt{|A|} / \sqrt{T}$.
\end{lemma}

\begin{proof}
The weight of iteration $t < T$ is $w_t = \Pi_{i = t}^{T-1} \frac{i^{3/2}}{i^{3/2}+1}$ and $w_T = 1$. Thus, $w_t \le 1$ for all $t$ and therefore $\sum_{t=1}^T w_t^2 \le T$.

Additionally, $w_t \ge \Pi_{i = t}^{T-1} \frac{i}{i+1} = \frac{t}{T}$ for $t < T$ and $w_T = 1$. Thus, $\sum_{t=1}^T w_t \ge T(T+1)/(2T) > T/2$.

Applying (\ref{eq:discount_bound}) and Lemma~\ref{le:neg}, we see that $Q_i^{w,T}(I,a) \le \frac{\Delta\sqrt{|A|}\sqrt{\sum_{t=1}^Tw_t^2}}{\sum_{t=1}^Tw_t} \le \frac{2\Delta\sqrt{|A|}\sqrt{T}}{T}$. From (\ref{eq:bound}) we see that $Q_i^{w,T} \le \frac{2\Delta|\mathcal{I}_i|\sqrt{|A|}\sqrt{T}}{T}$. Since $R_i^{w,T} \le Q_i^{w,T}$, so $R_i^{w,T} \le \frac{2\Delta|\mathcal{I}_i|\sqrt{|A|}\sqrt{T}}{T}$.
\end{proof}

\section{Correctness of DCFR(3/2, 1/2, 2)}

\begin{theorem}
	Assume that $T$ iterations of DCFR$_{\frac{3}{2},\frac{1}{2},2}$ are conducted in a two-player zero-sum game. Then the weighted average strategy profile is a $9\Delta|\mathcal{I}||\sqrt{|A|}/\sqrt{T}$-Nash equilibrium.
\end{theorem}

\begin{proof}
From Lemma~\ref{le:sqrtneg}, we know that regret in DCFR$_{\frac{3}{2},\frac{1}{2},2}$ for any infoset $I$ and action $a$ cannot be below $-\Delta\sqrt{T}$.

Consider the weighted sequence of iterates $\sigma'^1,...,\sigma'^T$ in which $\sigma'^t$ is identical to $\sigma^t$, but weighed by $w_{a,t} = \Pi_{i = t}^{T-1} \frac{i^2}{(i + 1)^2} = \frac{6t^2}{T(T+1)(2T+1)}$ rather than $w_t = \Pi_{i = t}^{T-1} \frac{i^{3/2}}{i^{3/2} + 1}$. The regret of action $a$ in infoset $I$ on iteration $t$ of this new sequence is $R'^t(I,a)$.

From Lemma~\ref{le:dcfr_regret} we know that $R^t(I,a) \le 2\Delta\sqrt{|A|}\sqrt{T}$ for player~$i$ for action $a$ in infoset $I$. Since $w_{a,t}$ is an increasing sequence, so we can apply Lemma~\ref{le:plausible} using weight $w_{a,t}$ for iteration $t$ with $B = 2\Delta\sqrt{|A|}\sqrt{T}$ and $C = -\Delta\sqrt{|A|}\sqrt{T}$. From Lemma~\ref{le:plausible}, this means that $R'^t(I,a) \le \frac{6 T^2 (3\Delta \sqrt{|A|} \sqrt{T})}{\big(T(T+1)(2T + 1)\big)} \le 9\Delta(|\sqrt{|A|})/\sqrt{T}$. Applying (\ref{eq:bound}), we get weighted regret is at most $9\Delta|\mathcal{I}_i|(|\sqrt{|A|})/\sqrt{T}$. Since the weights sum to one, this is also weighted average regret. Since $|\mathcal{I}_1| + |\mathcal{I}_2| = |\mathcal{I}|$, so the weighted average strategies form a $9\Delta(|\mathcal{I}|(\sqrt{|A|})/\sqrt{T}$-Nash equilibrium.
\end{proof}

\begin{lemma}
	\label{le:sqrtneg}
	Suppose after each of $T$ iterations of CFR, regret is multiplied by $\frac{\sqrt{t}}{\sqrt{t} + 1}$ on iteration $t$. Then $R^T(I,a) \ge -\Delta \sqrt{T}$ for any infoset $I$ and action $a$.
\end{lemma}

\begin{proof}
	We prove this inductively. On the first iteration, the lowest regret could be after multiplying by $\frac{\sqrt{1}}{\sqrt{1} + 1}$ is $-\frac{\Delta}{2}$. Now assume that after $T$ iterations of CFR in which regret is multiplied by $\frac{\sqrt{t}}{\sqrt{t}+1}$ on each iteration, $R^T(I,a) \ge -\Delta\sqrt{T}$ for infoset $I$ action $a$. After conducting an additional iteration of CFR and multiplying by $\frac{\sqrt{T + 1}}{\sqrt{T+ 1} + 1}$, $R^{T+1}(I,a) \le -\Delta(\sqrt{T} + 1)\frac{\sqrt{T+1}}{\sqrt{T+1} + 1}$.
	Since $\sqrt{T} + 1 \le \sqrt{T + 1} + 1$, so $-\frac{\sqrt{T} + 1}{\sqrt{T + 1} + 1} \Delta\sqrt{T + 1} = -\Delta(\sqrt{T} + 1) \frac{\sqrt{T+ 1}}{\sqrt{T+ 1} + 1} \ge -\Delta\sqrt{T + 1}$. Thus, $R^{T+1}(I,a) \ge -\Delta\sqrt{T+1}$.
\end{proof}

\end{document}